\documentclass[letterpaper, 10 pt, conference]{ieeeconf}  

\IEEEoverridecommandlockouts                              

\usepackage{amsmath} 
\usepackage{amssymb}  

\usepackage{graphicx}
\usepackage{hhline}
\usepackage{algorithmic}
\usepackage[ruled,vlined,titlenotnumbered]{algorithm2e}
\usepackage{caption,subcaption}
\usepackage{url}

\newcommand{\RR}{\mathbb{R}}

\newtheorem{definition}{Definition}
\newtheorem{theorem}{Theorem}
\newtheorem{lemma}{Lemma}

\newtheorem{example}{Example}
\newtheorem{proposition}{Proposition}

\DeclareMathOperator*{\argmax}{arg\,max}
\DeclareMathOperator*{\argmin}{arg\,min}

\def\longversion{1} 

\usepackage{tikz}
\usetikzlibrary{shapes.geometric} 
\usetikzlibrary{calc}
\usetikzlibrary{fit}
\usetikzlibrary{arrows,automata}

\tikzstyle{process} = [rectangle, minimum width=3cm, minimum height=1cm, text centered, text width=3cm, draw=black]
\tikzstyle{decision} = [diamond, minimum width=3cm, minimum height=1cm, text badly centered, draw=black]
\tikzstyle{cloud} = [draw, ellipse, node distance=3cm, minimum height=2em]
\tikzstyle{startstop} = [rectangle, rounded corners, minimum width=1cm, minimum height=1cm,text centered, draw=black]
\tikzstyle{arrow} = [thick,->,>=stealth]

\tikzstyle{block} = [draw, rectangle, 
    minimum height=3em, minimum width=2.9cm, text centered]
\tikzstyle{sum} = [draw, fill=blue!20, circle, node distance=1cm]
\tikzstyle{input} = [coordinate]
\tikzstyle{output} = [coordinate]
\tikzstyle{pinstyle} = [pin edge={to-,thin,black}]

\title{\LARGE \bf
Computing Complexity-aware Plans Using Kolmogorov Complexity
}

\author{Elis Stefansson$^{1}$ and Karl H. Johansson$^{1}$ 
\thanks{$^{1}$School of Electrical Engineering and Computer Science, KTH Royal Institute of Technology, Sweden. Email: \{elisst, kallej\}@kth.se}%
\thanks{This work was partially funded by the Swedish Foundation for Strategic Research, the Swedish Research Council, and the Knut och Alice Wallenberg foundation.}
}

\begin{document}

\maketitle
\thispagestyle{empty}
\pagestyle{empty}

\begin{abstract}
In this paper, we introduce complexity-aware planning for finite-horizon deterministic finite automata with rewards as outputs, based on Kolmogorov complexity. Kolmogorov complexity is considered since it can detect computational regularities of deterministic optimal policies. We present a planning objective yielding an explicit trade-off between a policy's performance and complexity. It is proven that maximising this objective is non-trivial in the sense that dynamic programming is infeasible. We present two algorithms obtaining low-complexity policies, where the first algorithm obtains a low-complexity optimal policy, and the second algorithm finds a policy maximising performance while maintaining local (stage-wise) complexity constraints. We evaluate the algorithms on a simple navigation task for a mobile robot, where our algorithms yield low-complexity policies that concur with intuition. 
\end{abstract}

\section{Introduction}

\subsection{Motivation}
Artificial intelligence has under the last decade progressed significantly achieving superhuman performance in challenging domains such as the video game Atari and the board game Go \cite{mnih2015human,silver2016mastering}. Unfortunately, for more complex and unconstrained environments (e.g., advanced real-world systems such as autonomous vehicles) results are more limited \cite{recht2018tour}. One major challenge here is the huge space of all possible strategies (policies) that the agent (i.e., robot or machine) can perform, making tractable solutions cumbersome to obtain naively. However, humans tend to face these complex domains with relative ease.

One explanation why humans perform well in complex tasks comes from cognitive neuroscience, proposing that general intelligence is linked to efficient compression, known as the efficient coding hypothesis \cite{attneave1954some, barlow1961possible, sims2018efficient}. This idea is not new but can be traced back to William of Occam saying \emph{``If there are alternative explanations for a phenomenon, then, all other things being equal, we should select the simplest one''}, where the simple alternative is the alternative with the shortest explanation. This methodology is known as Occam's razor \cite{li2008introduction}. Formalisations of Occam's razor have been able to detect computational regularities in for example the decimals of $\pi$ \cite{zenil2018decomposition} and automatically extract low-complexity physical laws (e.g., $E = {mv^2}/{2}$) directly from data \cite{udrescu2020ai}. In this paper, we are interested if a similar formalisation can automatically extract low-complexity policies, seen as a first step towards more tractable and intelligent behaviour for agents acting in complex environments.


\subsection{Contribution}
The main contribution of this paper is to define a complexity measure for policies in deterministic finite automata (DFA) \cite{Hopcroft2006book} using Kolmogorov complexity \cite{li2008introduction}, and to construct tractable complexity-aware planning algorithms with explicit trade-offs between performance and complexity based on this measure.
More precisely, our contributions are~three-fold:

Firstly, we define a complexity measure for deterministic policies in finite-horizon DFA with rewards as outputs. This measure uses Kolmogorov complexity to evaluate how complex a policy is to execute. Kolmogorov complexity, which can be seen as a formalisation of Occam's razor, is a computational notion of complexity being able to not only detect statistical regularities (commonly exploited in standard information theory) but also computational regularities.\footnote{An example of a sequence with computational regularity but no apparent statistical regularity is the infinite sequence $1234567891011 \dots$.} Our key insight is that optimal policies typically posses computational regularities (such as reaching a goal state) apart from statistical ones, making Kolmogorov complexity an appealing complexity evaluator.

Secondly, we present a complexity-aware planning objective based on our complexity measure, yielding an explicit trade-off between a policy's performance and complexity. We also prove that maximising this objective is non-trivial in the sense that dynamic programming \cite{Bellman:1957} is infeasible. 

Thirdly, we present two algorithms obtaining low-complexity policies. The first algorithm, Complexity-guided Optimal Policy Search (COPS), finds a policy with low complexity among all optimal policies, following a two-step procedure. In the first step, the algorithm finds all optimal policies, without any complexity constraints. Dropping the complexity constraints, this step can be done using dynamic programming. The second step runs a uniform-cost search \cite{aibook} over all optimal policies, guided by a complexity-driven cost, favouring low-complexity policies and enabling moderate search depths. The second algorithm, Stage-complexity-aware Program (SCAP), penalises instead policies locally for executing complex manoeuvres. This is done by partitioning the horizon into stages, with local complexity constraints over the stages instead of the full horizon, which enables dynamic programming over the stages. Finally, we evaluate our algorithms on a simple navigation task where a mobile robot tries to reach a certain goal state. Our algorithms yield low-complexity policies that concur with intuition.

\subsection{Related Work}
The interest of complexity in control and learning has a long history going back to Bellmann's  curse of dimensionality \cite{Bellman:1957}, Witsenhausen's counterexample \cite{witsenhausen1968counterexample, ho1972team} and the general open problem under what conditions LQG admits an optimal low-dimensional feedback controller \cite{bb5ed01f-5468-46ac-b45c-9d8e1b6fb224}, just to mention a few, and has lately been reviewed as an essential component for intelligent behaviour \cite{russell2016rationality}. Recent attempts to find low-complexity policies can be divided into two categories. In the first category, the system itself is approximated by a low-complexity system (e.g., smaller dimension), whereas an approximately optimal solution can be obtained. Methods in this category include bisimulation \cite{girard2011approximate,biza2020learning}, PCA analysis \cite{roy2003exponential,liu2017exploiting}, and information-theoretic compression such as the information bottleneck method \cite{abel2019state,larsson2020information}. In the second category, a low-complexity policy is instead obtained directly. Here, notable methods include policy distillation \cite{rusu2015policy}, VC-dimension constraints \cite{kearns2000approximate}, concise finite-state machine plans \cite{concise2017, pervanAlgorithmic2021}, low-memory policies through sparsity constraints \cite{pmlr_v144_booker21a}, and information-theoretic approaches such as KL-regularisation \cite{rubin2012trading,tishby2011information}, mutual information regularisation with variations \cite{tanaka2017lqg,fox2016minimum,tanaka2017transfer}, and minimal specification complexity \cite{1178902,4476333}. Our work belongs to this second category and resembles \cite{concise2017, pervanAlgorithmic2021,1178902,4476333} the most, but differ since we consider Kolmogorov~complexity.

Kolmogorov complexity has also been considered in the context of reinforcement learning as a tool for complexity-constrained inference \cite{cohen2019strongly,aslanides2017universal,hutter2004universal} based on Solomonoff's theory of inductive inference \cite{li2008introduction}. We differ by focusing instead on constraining the computational complexity of the obtained policy itself, assuming the underlying system to be~known. 

Finally, the work \cite{chmait2016factors} considers Kolmogorov complexity to measure the complexity of an action sequence similar to this line of work. We differ by also optimising over the complexity, while \cite{chmait2016factors} only evaluates the complexity of an immutable~object. 

\subsection{Outline}
The remaining paper is as follows. Section~\ref{problem_formulation} provides  preliminaries together with the problem statement. Section~\ref{planning} defines the complexity measure together with the complexity-aware planning objective, and proves that dynamic programming is infeasible. Section~\ref{algorithms} presents two algorithms for yielding low-complexity policies and Section~\ref{numerical} evaluates the algorithms on a mobile robot example. Finally, Section~\ref{conclusion} concludes the~paper.

\if\longversion0
An extended version of this paper can be found at \cite{stefansson2021cdc}.
\fi

\section{Problem Formulation}\label{problem_formulation}

\subsection{Turing Machines}
We give a brief review of Turing machines following \cite{li2008introduction}.
A Turing machine is a mathematical model of a computing device, manipulating symbols on an infinite list of cells (known as the tape), reading one symbol at a time on the tape with a one access pointer (known as the head). Formally:

\begin{definition}
A Turing machine $M$ is a tuple $(A,Q,\delta,q_0)$ with: A finite set of tape symbols $A \cup \{b \}$ with alphabet $A$ and blank symbol $b$; A finite set of states $Q$ with start state $q_0 \in Q$; A partial function\footnote{We use the notation $f: A \nrightarrow B$ to denote a partial function from a set $A$ to a set $B$ (i.e., a function that is only defined on a subset of $A$).} $\delta : Q \times (A \cup \{b\}) \nrightarrow (A \cup \{b\} \cup \{L,R \}) \times Q$ called the transition function.
\end{definition}
An execution of a Turing machine $M$ with input string\footnote{Here, $A^*$ denotes the set of all finite strings from the alphabet $A$, e.g., $x=010 \in A^*$ if $A=\{0,1\}$.} ${x \in A^*}$ starts with $x$ on the tape (and blank symbols on both sides of $x$), $q=q_0$ as state, and head on the first symbol $s$ of $x$. It then transitions according to $(h,p) = \delta(q,s)$, where $h \in A \cup \{b\}$ specifies what current scanned symbol $s$ should be replaced with, or if head should move left (right) one step ($h \in \{L,R \}$), and $p \in Q$ specifies the next state. This transition procedure is repeated until $\delta(q,s)$ becomes undefined. In this case, $M$ halts with output $y \in A^*$ equal to the maximal string in $A^*$ currently under scan, or $y=0$ if $b$ is scanned. For each Turing machine $M$, this input-output convention defines a partial function $\phi$ (defined when $M$ halts) known as a \emph{partial computable function}.

An important special class of Turing machines is the universal Turing machines. Informally, a universal Turing machine is a Turing machine that can simulate any other Turing machine, and can therefore be seen as an idealised version of a modern computer. A fundamental result states that there exist such machines \cite{li2008introduction}, a fact utilised when defining the Kolmogorov complexity.

\subsection{Kolmogorov Complexity}\label{Kolmogorov_complexity}
In this section, we present needed theory concerning Kolmogorov complexity  \cite{li2008introduction}. Informally, the Kolmogorov complexity of an object $x$ is the length of the smallest program $p$ on a computer that outputs $x$. If $p$ is much smaller than $x$, then we have compressed the information in $x$ to only the essential information of $x$.
The Kolmogorov complexity can therefore be seen as a formalisation of Occam's razor, stripping away all the non-essential information. More formally, the computer is a universal Turing machine $U$, and the object $x \in A^*$ and the program $p \in A^*$ are strings of a finite alphabet $A$. Towards a precise definition, we need the following notion:

\begin{definition}[\cite{li2008introduction}]
The complexity of a string $x \in A^*$ with respect to a partial computable function $f: A^* \nrightarrow A^*$ is defined as
\begin{equation*}
K_f(x) = \min \{\ell(p): f(p) = x \},
\end{equation*}
where $\ell(p)$ is the length of string $p$. Here, $p$ is interpreted as the program input to $f$, and $K_f(x)$ is thus the length of the smallest program that, via $f$, describes $x$. 
\end{definition}
The following fundamental result, known as the invariance theorem, asserts that there is a partial computable function with shorter (i.e., more compact) descriptions than any other partial computable function, up to a constant:
\begin{theorem}[Invariance theorem \cite{li2008introduction}]\label{invariance_theorem}
There exists a partial computable function $\phi$, constructed from a universal Turing machine $U$, with the following \emph{additively optimal} property: For any other partial computable function $\psi$, there exists a constant $c_\psi$ (dependent only on $\psi$) such that for all $x \in A^*$: $K_\phi(x) \leq K_\psi(x)+c_{\psi}$.
\end{theorem}
This function $\phi$ serves as our computer when defining the Kolmogorov complexity:
\begin{definition}[Kolmogorov Complexity  \cite{li2008introduction}]\label{definition_kolmogorov_complexity}
Let $\phi$ be as in Theorem \ref{invariance_theorem}. The Kolmogorov complexity of $x \in A^*$ is $K_\phi(x)$. We sometimes abbreviate $K_\phi$ as $K$. 
\end{definition}
The Kolmogorov complexity is robust in the sense that it depends only benignly on the choice of $\phi$, see \cite{li2008introduction} for details.
The Kolmogorov complexity $K(x)$ is however not computable in general, but can only be over-approximated pointwise \cite{li2008introduction}. The method we use in this paper to approximate $K(x)$ is from \cite{zenil2018decomposition} and based on algorithmic probability \cite{li2008introduction}, see\if\longversion0 \cite{stefansson2021cdc}
\else
\space Appendix
\fi
for a summary.
This estimation method is used in the simulations (Section \ref{numerical}), while the underlying theoretical framework we develop (Section \ref{planning} and \ref{algorithms}) is based on the exact Kolmogorov complexity.\footnote{We stress that other estimation methods can also be used, e.g., Lempel-Ziv compression \cite{lempel1976complexity} (e.g., used in \cite{chmait2016factors}), since the theoretical framework is independent on the particular estimation method chosen. We picked \cite{zenil2018decomposition} due to its more direct connection with Kolmogorov complexity, whereas Lempel-Ziv compression relies on classical information theory.}

\subsection{Deterministic Finite Automata}
We consider finite-horizon planning for discrete systems formalised as time-varying DFA \cite{Hopcroft2006book} with actions as inputs and rewards as outputs (i.e., time-varying Mealy machines \cite{Mealy1955}), on the form:\footnote{The time-varying feature of the DFA can be lifted by including the time into the state. We keep the current notation for easier readability.}
\begin{equation}\label{MDP_eq_tuple}
\mathcal{M} = \langle S,A,T, (f_t, r_t, F_t, t \in \mathbb{T} ) \rangle,
\end{equation}
where $S$ is the finite set of states, $A$ the finite action set, $T \in \mathbb{N}$ the horizon length and $\mathbb{T} = \{0,1,\dots,T\}$, $f_t: S \times A \rightarrow S$ the transition function specifying the next state $f_t(s,a) \in S$ at time $t$ given current state $s \in S$ and action $a \in A$, and $r_t: S \times A \rightarrow \RR$ the corresponding received reward $r_t(s,a) \in \RR$. The system stops at $t=T$, given by final state sets $F_T = S$ and $F_t=\emptyset$ for $t<T$. Given a start state $s_0 \in S$, the objective is to find a (deterministic) policy $\pi: \mathbb{T}\times S \rightarrow A$ maximising the total reward $\sum_{t=0}^T r_t(s_t,a_t)$ subject to the transition dynamics $s_{t+1} = f_t(s_t,\pi(t,s_t))$. A policy which does this for every start state $s_0 \in S$ is called an \emph{optimal}~policy.

\subsection{Problem Statement}
We now formalise the problem statement. In this paper, we answer the following questions: 
\begin{enumerate}
\item Given an $\mathcal{M}$ as in \eqref{MDP_eq_tuple}, how can one define a complexity measure $C(s_0,\pi)$ capturing how complex a policy $\pi$ is to execute from a start state $s_0 \in S$?
\item How can one construct a planning objective $\Psi$ with a formal trade-off between total reward $\sum_{t=0}^T r_t(s_t,a_t)$ and complexity $C(s_0,\pi)$ for a policy $\pi$?
\item How can one construct algorithms that maximises $\Psi$?
\item In particular, can dynamic programming be used to maximise~$\Psi$? 
\end{enumerate}
Questions 1, 2 and 4 are answered in Section \ref{planning}, while Section~\ref{algorithms} answers question 3 together with the numerical evaluations in Section \ref{numerical}.

\section{Complexity-aware Planning}\label{planning}
This section introduces a complexity measure in Section \ref{execution_complexity} and then sets up an appropriate complexity-aware planning objective in Section \ref{complexity_planning_obj}. Finally, we prove that dynamic programming cannot be used to maximise this objective in Section \ref{infeasible_dp}. Throughout this section, we fix a system $\mathcal{M}$ as in \eqref{MDP_eq_tuple}.

\subsection{Execution Complexity}\label{execution_complexity}
We start by defining our complexity measure for policies, capturing how complex it is to execute a policy $\pi$ from a start state $s_0 \in S$. Towards this, note that, given a start state $s_0 \in S$ and a policy $\pi$, we get a sequence of actions $(\pi(0,s_0), \dots, \pi(T,s_T))$ in $A^{T+1}$ from time $t=0$ to time $t=T$, generated by $\pi$ and the dynamics $s_{t+1} = f_t(s_t,\pi(t,s_t))$. We denote this action sequence by $E(s_0,\pi)$ and say that $\pi$ has low \emph{execution complexity} if $K(E(s_0,\pi))$ is low. Intuitively, a policy $\pi$ with low execution complexity has a small program that can execute it.

\begin{definition}[Execution Complexity]\label{execution_complexity_def}
Given a start state $s_0 \in S$, the execution complexity of a policy $\pi$ is ${C(s_0,\pi) = K(E(s_0,\pi))}$.
\end{definition}

\subsection{Complexity-aware Planning Objective}\label{complexity_planning_obj}
The execution complexity can be used to find policies with high total reward while keeping a low complexity. Formally, we want to find an action sequence $a_{0:T} =(a_0,a_1,\dots,a_T)$ maximising the objective\footnote{For convenience, we maximise directly over control inputs $a_t = \pi(t,s_t)$ instead of policies, since the horizon $T$ is typically lower than the number of states in $S$. With this convention, $K( a_{0:T})$ in \eqref{eq:alg_ec} agrees with $C(s_0,\pi)$ in Definition \ref{execution_complexity_def}.}
\begin{align}\label{eq:alg_ec}
\max_{a_{0:T}} \; \left [ \sum_{t=0}^T r_t(s_t,a_t)-\beta K( a_{0:T}) \right ],
\end{align}
subject to the dynamics $s_{t+1} = f_t(s_t,a_t)$ and start state $s_0 \in S$. Here, $\beta \geq 0$ determines how much we penalise complexity relative to obtaining a high total reward.

\begin{example}[Simple optimal policies]\label{ex_simple}
For $\beta>0$ sufficiently low, we obtain, for a start state $s_0 \in S$, an interesting subset of all optimal policies (i.e., policies maximising the total reward $\sum_{t=0}^T r_t(s_t,a_t)$) with the lowest complexity, which we here call \emph{simple optimal policies}. Note that, the simple optimal policies can also be found via the objective
\begin{equation}\label{eq:simple_optimal_policies}
\min_{a_{0:T}} K( a_{0:T}) \quad s.t. \; \; a_{0:T} \in \argmax_{a_{0:T}} \sum_{t=0}^T r_t(s_t,a_t),
\end{equation}
subject to the dynamics $s_{t+1} = f_t(s_t,a_t)$ and start state $s_0$. See\if\longversion0 \cite{stefansson2021cdc}
\else
\space Appendix
\fi
for a proof.
\end{example}

\subsection{Dynamic Programming is Infeasible}\label{infeasible_dp}

We are interested in algoritms obtaining the maximum of~\eqref{eq:alg_ec}. A standard method for finding an optimal action sequence maximising the total reward $\sum_{t=0}^T r_t(s_t,a_t)$ is dynamic programming \cite{Bellman:1957}. Hence, one may ask if \eqref{eq:alg_ec} can be solved using dynamic programming. The answer is negative and follows from the following result. 

\begin{proposition}\label{co:no_dyn_prog}
For sufficiently large $T \in \mathbb{N}$, we cannot decompose $K$ on the form
\begin{equation*}
K(x_{1:T}) = h_1(x_1)+g_{2:T}(x_{2:T})
\end{equation*}
where $ \{ g_{k:T} \}_{k=2}^{T}$ are functions given by
\begin{equation*}
g_{k:T}(x_{k:T}) = h_k(x_k)+g_{k+1:T}(x_{k+1:T})
\end{equation*}
for some functions $\{ h_i \}_{i=1}^T$. That is, we cannot decompose $K$ recursively into an immediate complexity plus a future complexity. In particular, we cannot obtain $K(x_{1:T})$ by backward induction.\footnote{Note that $h_i$ and $g_{i:T}$ can by \emph{any} functions, not necessarily computable. If one restricts these functions to be computable, then the result follows easily from the fact that $K$ itself is not commutable.}
\end{proposition}
See\if\longversion0 \cite{stefansson2021cdc}
\else
\space Appendix
\fi
for proof. Due to Proposition \ref{co:no_dyn_prog}, we cannot naively apply dynamic programming to solve \eqref{eq:alg_ec} (for large enough $T$). More precisely, the objective function
\begin{equation*}
\Psi(a_{0:T}) := \sum_{t=0}^T r_t(s_t,a_t)-\beta K( a_{0:T}),
\end{equation*}
subject to the dynamics $s_{t+1} = f_t(s_t,a_t)$ and start state $s_0 \in S$, cannot be decomposed on the form $\Psi(a_{k:T}) = \Psi_0(a_{0})+\Phi_{1:T}(a_{1:T})$, where $ \{ \Phi_{k:T} \}_{k=1}^{T}$ are given recursively by $\Phi_{k:T}(x_{k:T}) = \Psi_k(x_k)+\Phi_{k+1:T}(x_{k+1:T})$ for some functions $\{ \Psi_k \}_{i=0}^T$. Indeed, if such a decomposition existed, it would contradict Proposition \ref{co:no_dyn_prog}. Thus, $\Psi(a_{0:T})$ cannot be maximised using dynamic programming. Fortunately, there are methods that circumvent this issue, presented next.


\section{Complexity-aware Planning Algorithms}\label{algorithms}
This section answers question 3 in the problem statement, presenting two algorithms that circumvent the dynamic programming issue given by Section \ref{infeasible_dp}. The first algorithm, COPS, restricts the task to find simple optimal policies, while the second algorithm, SCAP, modifies the objective focusing on local (stage-wise) complexity.
 
\subsection{Complexity-guided Optimal Policy Search (COPS)}\label{alg:search}
COPS seeks a simple optimal policy as in Example \ref{ex_simple} by a two-step procedure. Step 1 conducts ordinary dynamic programming (maximising the total reward $\sum_{t=0}^T r_t(s_t,a_t)$). This yields a mapping $\Pi: \mathbb{T} \times S \rightarrow \mathcal{P}(A)$ such that $\Pi(t,s) \subseteq A$ are the optimal actions at time $t  \in \mathbb{T}$ and state $s \in S$.\footnote{Here, $\mathcal{P}(A)$ is the power set of $A$, i.e., the set of all subsets of $A$.} In step 2, a uniform-cost search \cite{aibook} is executed to find an optimal action sequence $a_{0:T}$ with low execution complexity. A node $n$ in this search is on the form $n = (t,s_t, a_{0:t-1}) \in \mathbb{T} \times S \times A^{t}$,
where $t$ and $s_t$ are the current time and state and $a_{0:t-1}$ is the sequence of previous actions taken to arrive at state $s_t$ at time $t$. If node $n$ is not a terminal node (i.e., $t \leq T$), the children of $n$ are given by all $(t+1,s_{t+1}, a_{0:t})$ such that $a_t \in \Pi(t,s_t)$ and $s_{t+1} = f_t(s_t,a_t)$, i.e., we expand only over optimal actions. The cost for a node $n = (t,s_t, a_{0:t-1})$ is set to $c(n) = K(a_{0:t-1})$. 
The heuristic intuition behind the cost is that a low-complexity sequence $a_{0:T}$ should be more likely to have low-complexity subsequences $a_{0:t-1}$ and, thus, the cost focuses the search on low-complexity sequences, enabling moderate search depths.

The uniform-cost search is conducted by iteratively generating children of the node $n$ with the lowest cost, starting from the root node $n_0 = (0,s_0, \emptyset)$. We terminate the search when a terminal node $n = (T+1,s_{T+1}, a_{0:T})$ has the lowest cost and return its action sequence $a_{0:T}$. The algorithm is summarised by Algorithm \ref{Alg:SOPA}. The termination is motivated by the following result:

\begin{proposition}\label{th:search}
Assume $K(a_{0:t-1}) \leq K(a_{0:t})$ holds for all ${t  \leq T}$ and optimal action sequences $a_{0:T}$. Then Algoritm~\ref{Alg:SOPA} returns an optimal action sequence $a_{0:T}$ with lowest execution complexity, i.e., $a_{0:T}$ maximises~\eqref{eq:simple_optimal_policies}.
\end{proposition}

We stress that $K(a_{0:t-1}) \leq K(a_{0:t})$ is not always true, since adding $a_t$ to $a_{0:t-1}$ may result in higher regularity than $a_{0:t-1}$ has alone. However, $K(a_{0:t-1}) \leq K(a_{0:t})$ is more anticipated since $K$ is on average an increasing function with respect to sequence length \cite{li2008introduction}, motivating the assumption in Proposition \ref{th:search}, and why the algorithm can work well in practice. Proposition \ref{th:search} follows readily by applying uniform-cost search properties, see\if\longversion0 \cite{stefansson2021cdc}.
\else
\space Appendix.
\fi


Efficient searching in Algorithm 1 is possible for moderate horizon lengths, demonstrated numerically in Section \ref{case_study_1}. However, for longer horizons, the procedure becomes intractable. The next algorithm works for longer horizons by focusing on local (stage-wise) complexity instead of the complexity of the whole action sequence. COPS can be seen as a special case of this latter algorithm with only one stage and $\beta_0>0$, defined below in \eqref{eq_modified_objective}, sufficiently low.

\SetKwFor{Loop}{loop}{do}{}

\begin{algorithm}
\footnotesize
\SetKwInOut{Input}{Input}
\Input{System $\mathcal{M}$ as in \eqref{MDP_eq_tuple} and start state $s_0 \in S$.}
\SetKwInOut{Output}{Output}
\Output{Low-complexity optimal action sequence $a_{0:T}$.}
 \textbf{Step 1: Dynamic programming} \\
 $V_{T+1}(s) = 0$ for all $s \in S$\;
 \For{$t=T,T-1,\dots,0$}{
 \ForAll{$s \in S$}{
 \ForAll{$a \in A$}{
 $Q_t(s,a) := r_t(s,a)+V_{t+1}(f_t(s,a))$\;
 }
 $V_t(s) = \max_{a \in A} Q_t(s,a)$\;
 $\Pi(t,s) = \argmax_{a \in A} Q_t(s,a)$\;
 }
 }
 \textbf{Step 2: Uniform-cost search} \\
 $q =\{ \}$ \quad \# priority queue ordered by the cost $c$\;
 Append root node $n_0 = (0,s_0,\emptyset)$ to $q$\;
 \Loop{}{
 Pop first node $n = (t,s_t, a_{0:t-1})$ from $q$\;
 \If{t=T+1}{
 \textbf{return} $a_{0:t-1}$\;
 }
 \ForAll{$a_t \in \Pi(t,s_t)$}{
 Append node $n=(t+1,f_t(s_t,a_t), a_{0:t})$ to $q$\;
 }
 }
 \caption{COPS}
 \label{Alg:SOPA}
\end{algorithm}

\subsection{Stage-Complexity-Aware Program (SCAP)}
SCAP modifies the objective to focus only on local (stage-wise) complexity. More precisely, the modified objective is set to
\begin{equation}\label{eq_modified_objective}
\max_{a_{0:T}} \; \left [ \sum_{t=0}^T r_t(s_t,a_t)-\sum_{k=0}^K \beta_k K(\bold{a}_k)\right ]. 
\end{equation}
Here, $K+1$ is the number of stages in the horizon partition, each with stage length $l$, and $\bold{a}_k := a_{lk:lk+l-1} \in A^l$ is the executed action sequence at stage $k \in \{0,\dots,K\}$, penalised by its execution complexity $K(\bold{a}_k)$ with weight $\beta_k \geq 0$.



The key insight is that the modified objective enables dynamic programming over the stages. More precisely, initialise the value function as $V_{K+1}(s) = 0$ for all $s \in S$ and obtain the value function for the remaining stages using the backward recursion
\begin{align}\label{eq:local_dp}
V_{k}(s_{lk}) = \max_{\bold{a}_k \in A^l} [ r^l_k(s_{lk},\bold{a}_k)-\beta_k K(\bold{a}_k)+\nonumber \\ V_{k+1}(f^l_k(s_{lk},\bold{a}_k)) ].
\end{align}
Here, $f^l_k(s_{lk},\bold{a}_k)$ denotes the state one arrives at by sequentially applying $\bold{a}_k$ starting from $s_{lk}$, and $r^l_k(s_{lk},\bold{a}_k) = \sum_{t=lk}^{lk+l-1} r_t(s_t,a_t)$. Provided $l$ is small enough, the maximisation in \eqref{eq:local_dp} can be done by going through all $\bold{a}_k \in A^l$. 

\subsubsection{Hard-constrained version} An alternative to \eqref{eq_modified_objective} is the hard-constraint objective
\begin{equation}\label{eq:hard_con_version}
\max_{a_{0:T}} \; \sum_{t=0}^T r_t(s_t,a_t), \; \; \text{s.t.} \; \; K(\bold{a}_k) \leq L_k, \; \forall k 
\end{equation}
for some constants $L_k \geq 0$. In this case, we conduct dynamic programming with $V_{K+1}(s) = 0$ for all $s \in S$ and backward~recursion
\begin{align}\label{eq:local_dp_hard}
V_{k}(s_{lk}) = \max_{K(\bold{a}_k) \leq L_k} \Big [ r^l_k(s_{lk},\bold{a}_k)+V_{k+1}(f^l_k(s_{lk},\bold{a}_k)) \Big ].
\end{align}
Using \eqref{eq:local_dp_hard} enables more efficient computations than \eqref{eq:local_dp} since the sequences $\mathcal{A}_k := \{ \bold{a}_k \in A^l : K(\bold{a}_k) \leq L_k \}$ are typically only a fraction of $A^l$. Moreover, $\mathcal{A}_k$ can be computed beforehand by going through all sequences in $A^l$ (tractable for small $l$), or sought using a uniform-cost search (for moderate $l$, see\if\longversion0 \cite{stefansson2021cdc}
\else
\space Appendix
\fi for details).



\subsubsection{Action sequence extraction}
Once $V_k$ has been computed using \eqref{eq:local_dp} or \eqref{eq:local_dp_hard}, $a_{0:T}$ maximising \eqref{eq_modified_objective} or \eqref{eq:hard_con_version} can be readily obtained, for a given start state $s_0 \in S$, by forward simulation over the stages (see\if\longversion0 \cite{stefansson2021cdc}
\else
\space Appendix
\fi for details). 

Algorithm \ref{Alg:local_dp} summarises the procedure for the hard-constrained version. 

\begin{algorithm}
\footnotesize
\SetKwInOut{Input}{Input}
\Input{System $\mathcal{M}$ as in \eqref{MDP_eq_tuple} and start state $s_0 \in S$.}
\SetKwInOut{Output}{Output}
\Output{Action sequence $a_{0:T}$ maximising \eqref{eq:hard_con_version}.} 
\textbf{Step 1: Dynamic programming} \\
 Set $V_{T+1}(s) = 0$ for all $s \in S$\;
 Compute $\mathcal{A}_k = \{ a_{0:l-1} \in A^l : K(a_{0:l-1}) \leq L_k \}$ for $k=0,1,\dots,K$ \;
 \For{$k=K,K-1,\dots,0$}{
 \ForAll{$s \in S$}{
 \ForAll{$a_{0:l-1} \in \mathcal{A}_k$}{
 $Q_k(s,a_{0:l-1}) = r^l_k(s,a_{0:l-1})+V_{k+1}(f^l_k(s,a_{0:l-1}))$\;
 }
 $V_k(s) = \max_{a_{0:l-1}} Q_k(s,a_{0:l-1})$\;
 }
 }
 \textbf{Step 2: Action sequence extraction} \\
 \For{$k=0,1,\dots,K$}{
  Set $a_{lk:lk+l-1} \in \argmax_{a_{0:l-1} \in \mathcal{A}_k} Q_k(s_{lk},a_{0:l-1})$ \;
  $s_{l(k+1)} = f^l_k(s_{lk},a_{lk:lk+l-1})$ \;
 \textbf{return} $a_{0:T}$
 }
 \caption{SCAP}
 \label{Alg:local_dp}
\end{algorithm}


\section{Numerical Evaluations}\label{numerical}
This section presents case studies evaluating the complexity-aware planning algorithms proposed in Section~\ref{algorithms}. 
We describe the test environment in Section \ref{test_env}, evaluate COPS in Section \ref{case_study_1} and then SCAP in Section \ref{case_study_2}. 
The Kolmogorov complexity is estimated using the method from~\cite{zenil2018decomposition}.

\begin{figure*}[t!]
\centering
\begin{subfigure}[b]{0.30\textwidth}
	\includegraphics[width=\textwidth]{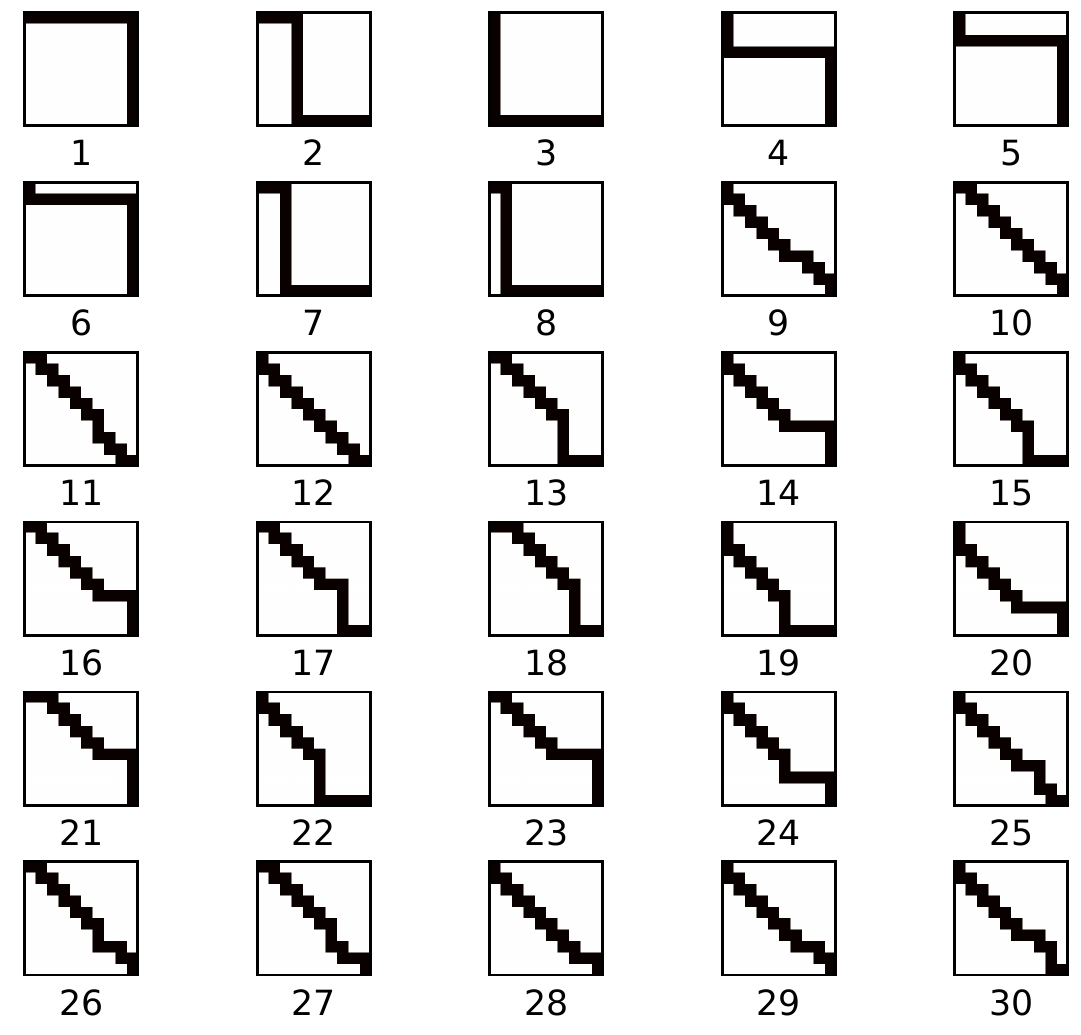}
    \caption{}
\label{fig:test1_fig_rec_side10}
\end{subfigure} 
\hfill
\begin{subfigure}[b]{0.30\textwidth}
	\includegraphics[width=\textwidth]{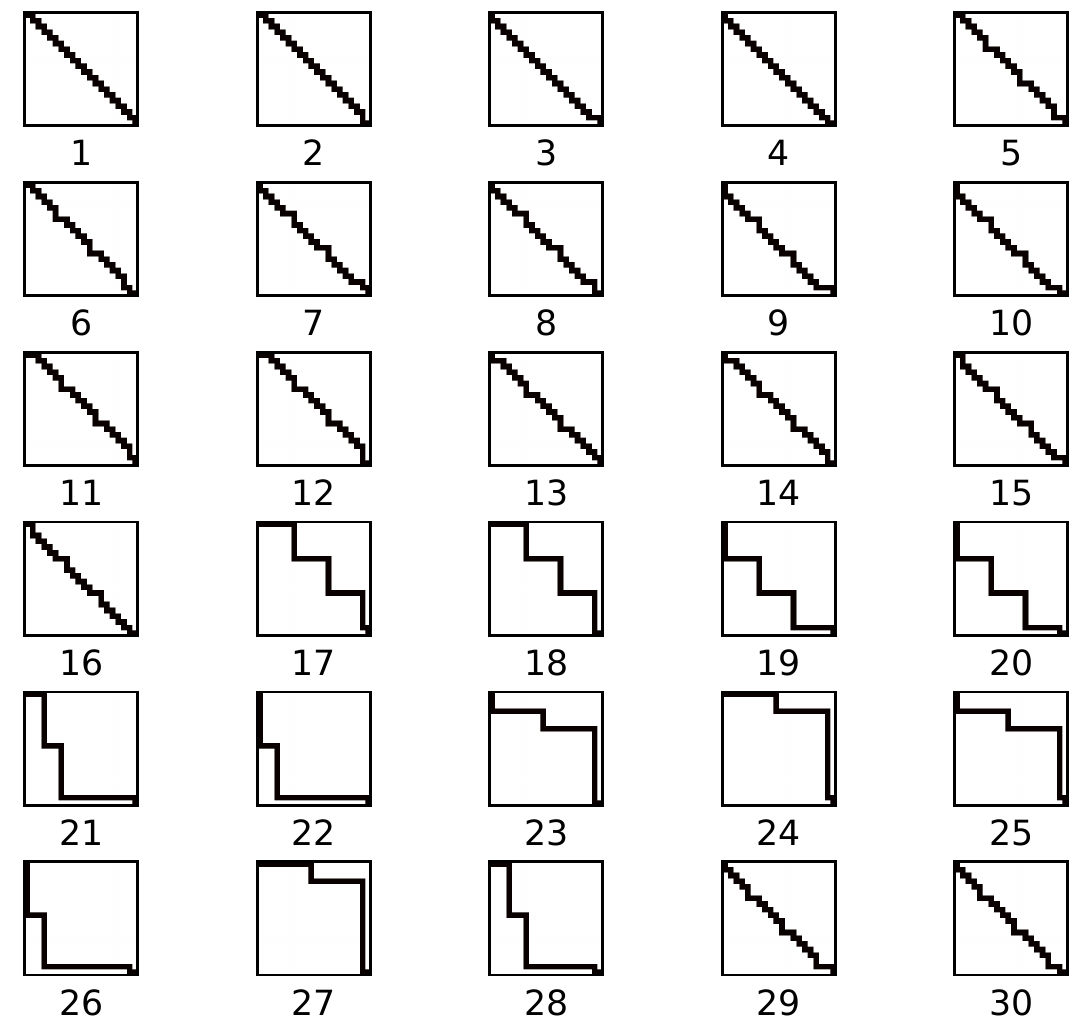}
    \caption{}
\label{fig:test1_fig_rec_side20}
\end{subfigure} 
\hfill
\begin{subfigure}[b]{0.30\textwidth}
	\includegraphics[width=\textwidth]{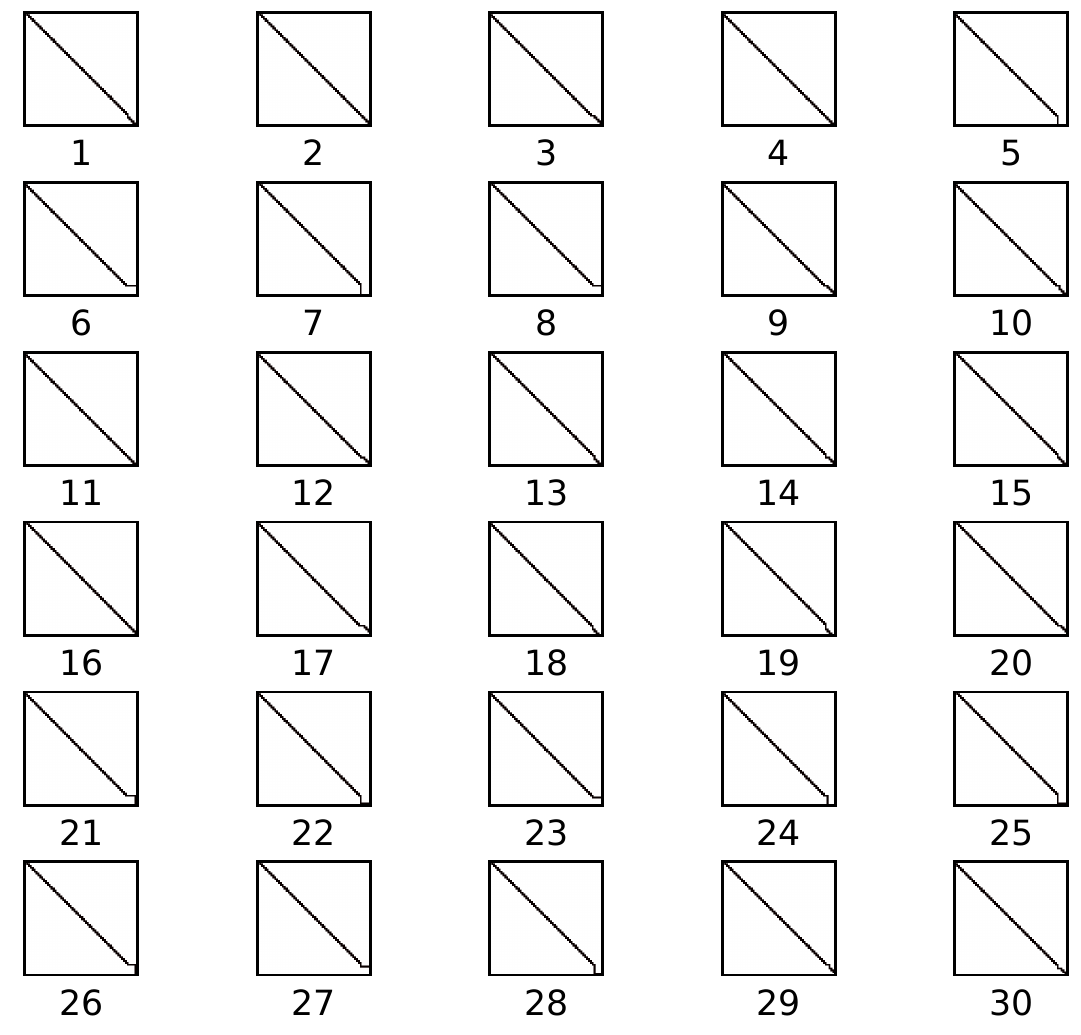}
    \caption{}
\label{fig:test1_fig_rec_side60}
\end{subfigure} 
\caption{Trajectories of the first 30 action sequences using COPS for: (a) small room, (b) medium room, and (c) large room.}
\label{fig:trajs}
\end{figure*}

\begin{table*}[!htb]
\centering
    \caption{Execution complexity for the found action sequences corresponding to Figure \ref{fig:trajs}.}
    \begin{subtable}{.3\linewidth}
      \centering
        \caption{}
\label{table_1}
\begin{tabular}{ll}
\hline
Action sequence & Execution complexity \\ \hline
1-4               & 47.30   \\
5-8               & 47.79     \\
9-12               & 47.91     \\
13-16               & 47.92     \\
17-24               & 48.30     \\
25-30               & 48.36     \\
\hline  
\end{tabular}
    \end{subtable}%
    \hspace{2.5em}
    \begin{subtable}{.3\linewidth}
      \centering
        \caption{}
\label{table_2}
\begin{tabular}{ll}
\hline
Action sequence & Execution complexity \\ \hline
1-4               & 36.49     \\
5-8               & 38.39     \\
9-12               & 38.79     \\
13-16               & 38.39     \\
17-20               & 38.60     \\
21-28               & 62.57     \\
29-30               & 39.75 \\ \hline
\end{tabular}
    \end{subtable} 
    \hspace{2.5em}
        \begin{subtable}{.3\linewidth}
      \centering
        \caption{}
\label{table_3}
\begin{tabular}{ll}
\hline
Action sequence & Execution complexity \\ \hline
1-4               & 58.80   \\
5-8               & 59.96   \\
9-16               & 60.53   \\
17-20               & 60.88   \\
21-28             & 61.20   \\
29-30             & 61.64   \\
\hline  
\end{tabular}
    \end{subtable} 
\end{table*}

\subsection{Test Environment}\label{test_env}
As test environment, we consider a system $\mathcal{M}$ as in \eqref{MDP_eq_tuple} where a robot (the agent) moves around in a room with discrete steps. More precisely, the room is a square with coordinates $N = \{1,,\dots,n\}$ in each direction forming the state space $S = N \times N$. At each time, the robot can move to any of the neighbouring spots or stay, i.e., the actions are $A = \{(\pm 1,0),(0,\pm1),(0,0)\}$
with dynamics $s_{t+1} = s_t+a_t$. If the robot executes an action that would move it outside the room, then it stays at the same spot. The reward is zero everywhere except for all state-action pairs $(s,a)$ that takes the robot to a given goal state $s^\star$, i.e., $f(s,a) = s^\star$, with reward $r(s,a)=1$. We set the goal state $s^\star=(n,n)$ to one of the corners of the room. The robot starts at $s_0 = (1,1)$ and the horizon $T=2(n-1)-1$ is set so that the agent precisely receives a reward for reaching $s^\star$ if acting optimally.

\subsection{Evaluation 1}\label{case_study_1}
\subsubsection{Setup} 
We first consider COPS. To see how  output and calculation time changes with problem size, we study the test environment $\mathcal{M}$ with different room sizes $n=10,20$ and $60$, called the small, medium and large room, respectively. We run the search until it has found 30 low-complexity action sequences, by keep expanding the node tree. Concretely, this is done by modifying the if-statement in Algorithm 1 to a condition appending every found action sequence $a_{0:T}$ into a list until this list has 30 sequences, and then return the list.

\subsubsection{Result} 
The running time for the three rooms are 21 seconds, 16 minutes and 15 minutes, respectively, and corresponding trajectories to the found action sequences are given in Fig. \ref{fig:test1_fig_rec_side10}, \ref{fig:test1_fig_rec_side20} and \ref{fig:test1_fig_rec_side60}, with execution complexities in Tables \ref{table_1}, \ref{table_2}, and \ref{table_3}. 

As a first example, consider the trajectory of the first action sequence found in the small room, labeled 1 in Fig. \ref{fig:test1_fig_rec_side10}. Here, the robot goes right until it reaches the upper right corner and then goes down to $s^\star$. That is, it exploits executing the same actions in batches to lower complexity, reaching an execution complexity of 47.30 seen in Table~\ref{table_1}. We also note that this is the same complexity as action sequences 2-4 in Fig. \ref{fig:test1_fig_rec_side10}. That 1 and 3 have the same complexity is due to symmetry, and the same is true for 2 and 4. However, why e.g., 1 and 2 have the same complexity (apart from the intuition that they both look like low-complexity executions) is unclear. It could be an inherited feature from the Kolmogorov complexity, or a bias from the estimation method, see\if\longversion0 \cite{stefansson2021cdc}
\else
\space Appendix
\fi for a discussion of the latter.

Looking now at all sequences in all three cases, we see that it is in general common to find sequences executing batches of the same actions and then alternate between such batches to lower execution complexity (1-8 in Fig. \ref{fig:test1_fig_rec_side10}, and 17-28 in Fig. \ref{fig:test1_fig_rec_side20}). Another typical feature to lower complexity is to exploit 2-periodic alternation between going down and going right (9-30 in Fig. \ref{fig:test1_fig_rec_side10}, 1-16 and 29-30 in Fig. \ref{fig:test1_fig_rec_side20}, and 1-30 in Fig. \ref{fig:test1_fig_rec_side60}). Also, what type of execution with lowest complexity varies with the room size, i.e., is horizon-dependent. Once again, this could come from the Kolmogorov complexity itself, where varying sequence length may facilitate different compression techniques, but could also originate from the estimation method.

We also note in Tables \ref{table_1} to \ref{table_3} that the algorithm mainly finds sequences in increasing complexity-order as anticipated by Proposition \ref{th:search}. Indeed, only the medium room case causes some unordered sequences, where notably the search finds the higher complexity sequences 21-28 with execution complexity 62.57 before reaching sequences 29-30 with complexity 39.75. This detour, caused by a violation in the assumption in Proposition \ref{th:search}, explains the long search time the medium room yields, longer than the large room although the latter has a larger horizon.

Finally, the cost $c(n) = K(a_{0:t-1})$ enables moderate horizon lengths $T$. In particular, for the large room, the number of possible optimal action sequences is around $2^{(T+1)} \approx 10^{36}$ after dynamic programming. However, the search only iterates $10^6$ nodes to find the first low-complexity sequence, a significant decrease due to the complexity-guiding cost.

\begin{figure*}[]
    \centering 
\begin{subfigure}{0.24\textwidth}
  \includegraphics[width=\linewidth]{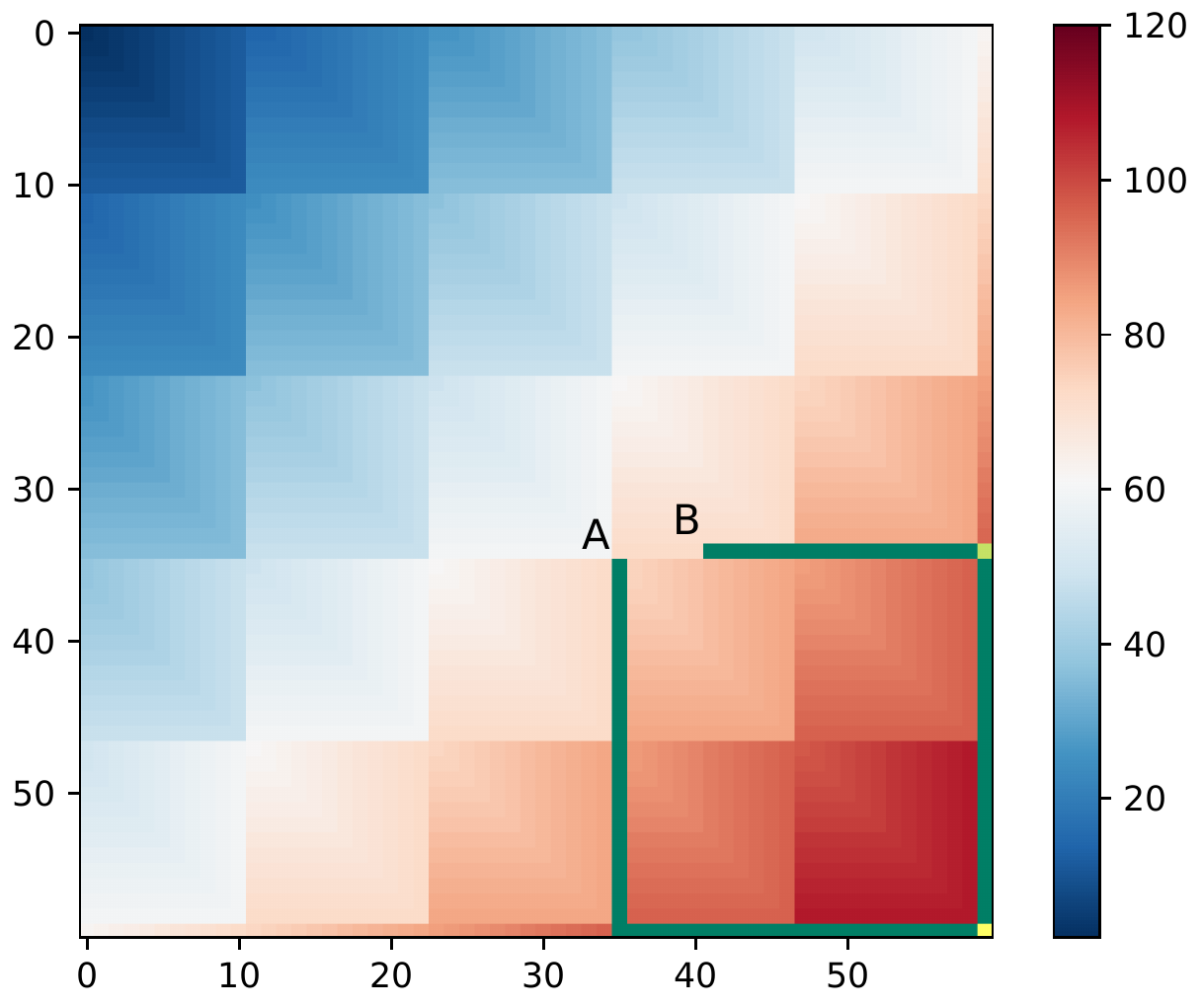}
  \caption{$L=26$}
  \label{fig:1}
\end{subfigure}\hfil 
\begin{subfigure}{0.24\textwidth}
  \includegraphics[width=\linewidth]{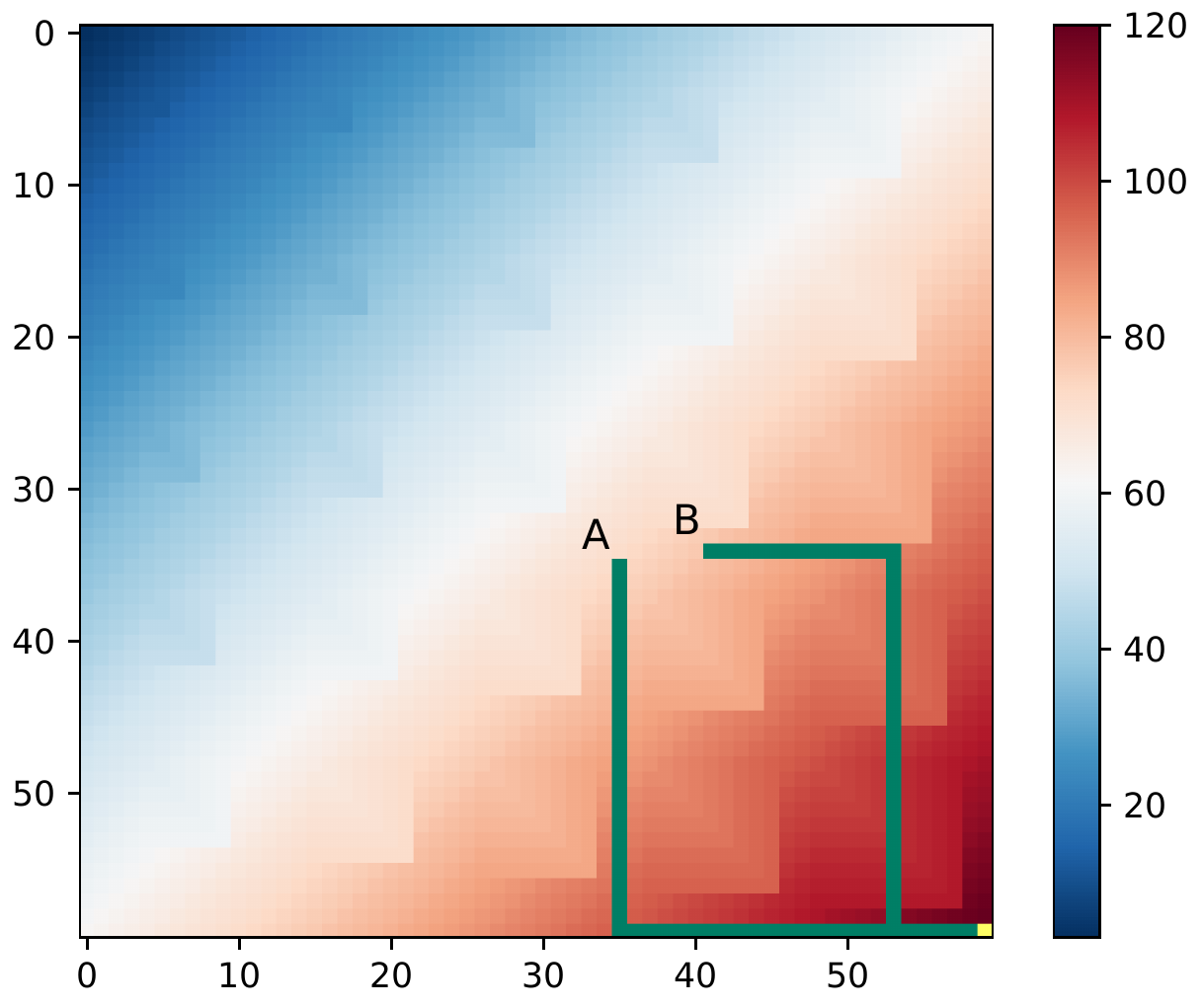}
  \caption{$L=28$}
  \label{fig:2}
\end{subfigure}\hfil 
\begin{subfigure}{0.24\textwidth}
  \includegraphics[width=\linewidth]{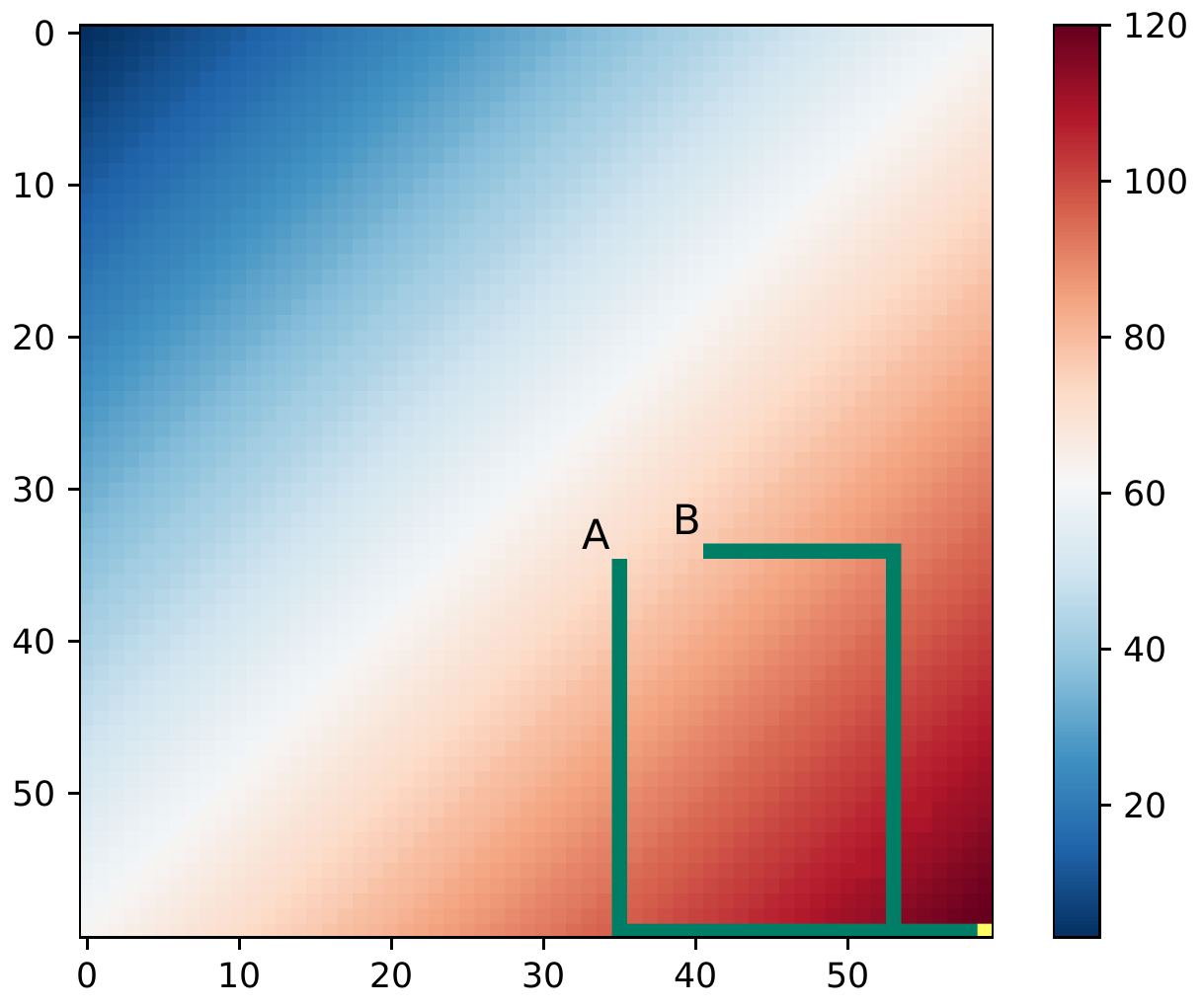}
  \caption{$L=30$}
  \label{fig:1}
\end{subfigure}\hfil 
\begin{subfigure}{0.24\textwidth}
  \includegraphics[width=\linewidth]{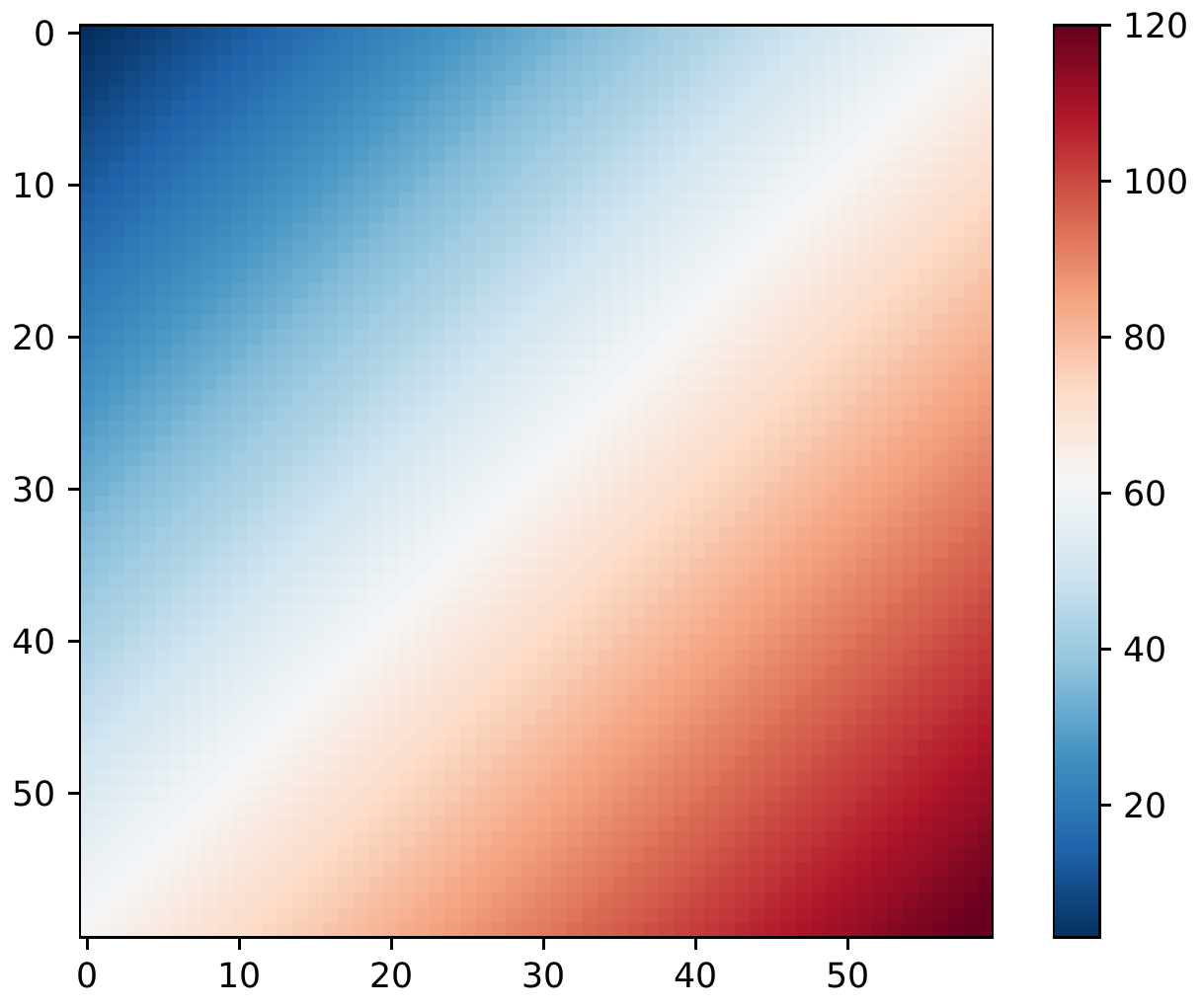}
  \caption{$L=\infty$}
  \label{fig:2}
\end{subfigure}
\caption{$V_0(s)$ for Setup 1 in Section \ref{case_study_2} and $L \in \{26,28,30,\infty\}$. Shown are also two corresponding trajectories (for each finite $L$) starting from A and B, obtained by Algorithm \ref{Alg:local_dp}.}
\label{fig:value_functions_corner}
\end{figure*}

\begin{figure*}[]
    \centering 
\begin{subfigure}{0.24\textwidth}
  \includegraphics[width=\linewidth]{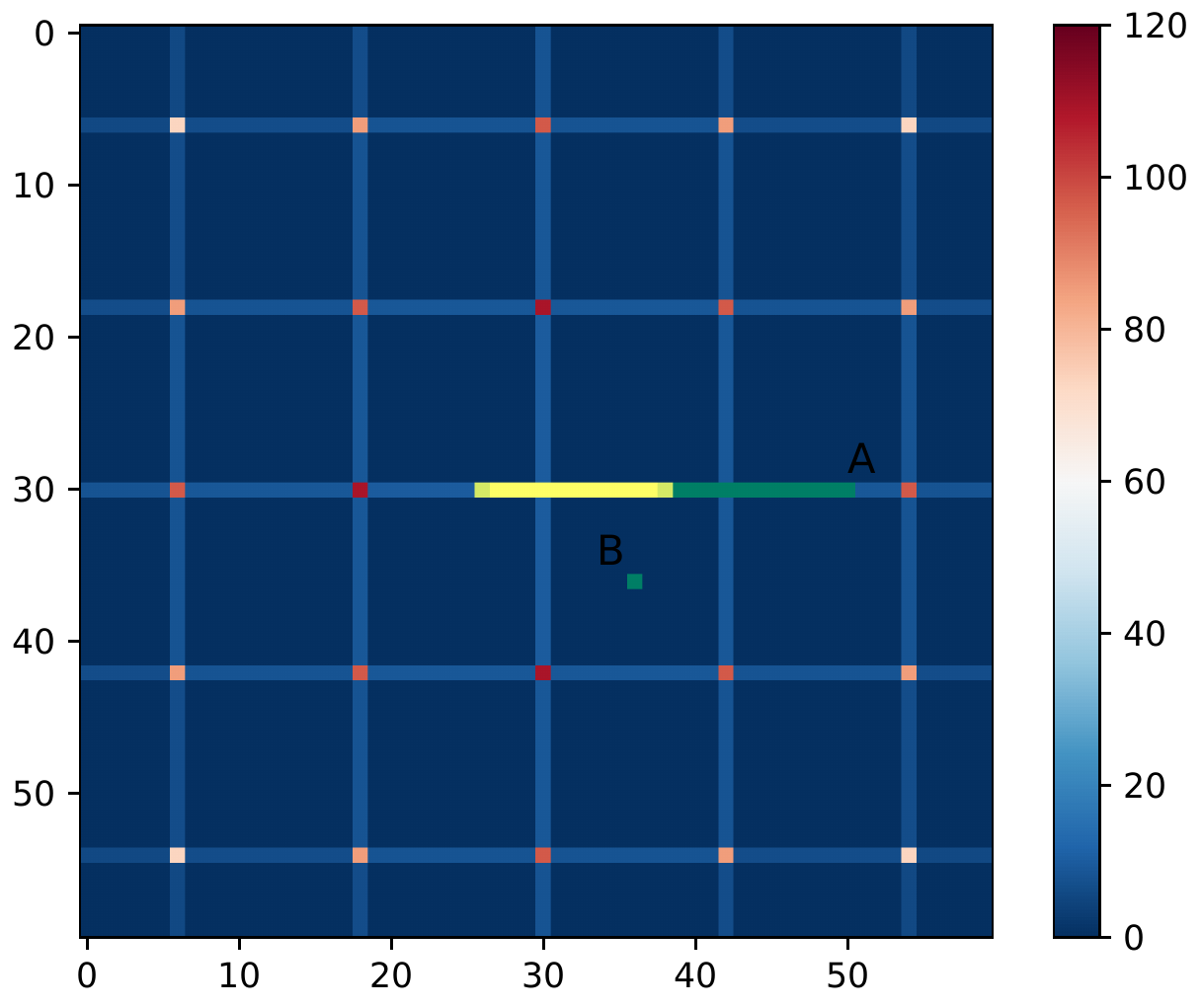}
  \caption{$L=26$}
  \label{fig:1}
\end{subfigure}\hfil 
\begin{subfigure}{0.24\textwidth}
  \includegraphics[width=\linewidth]{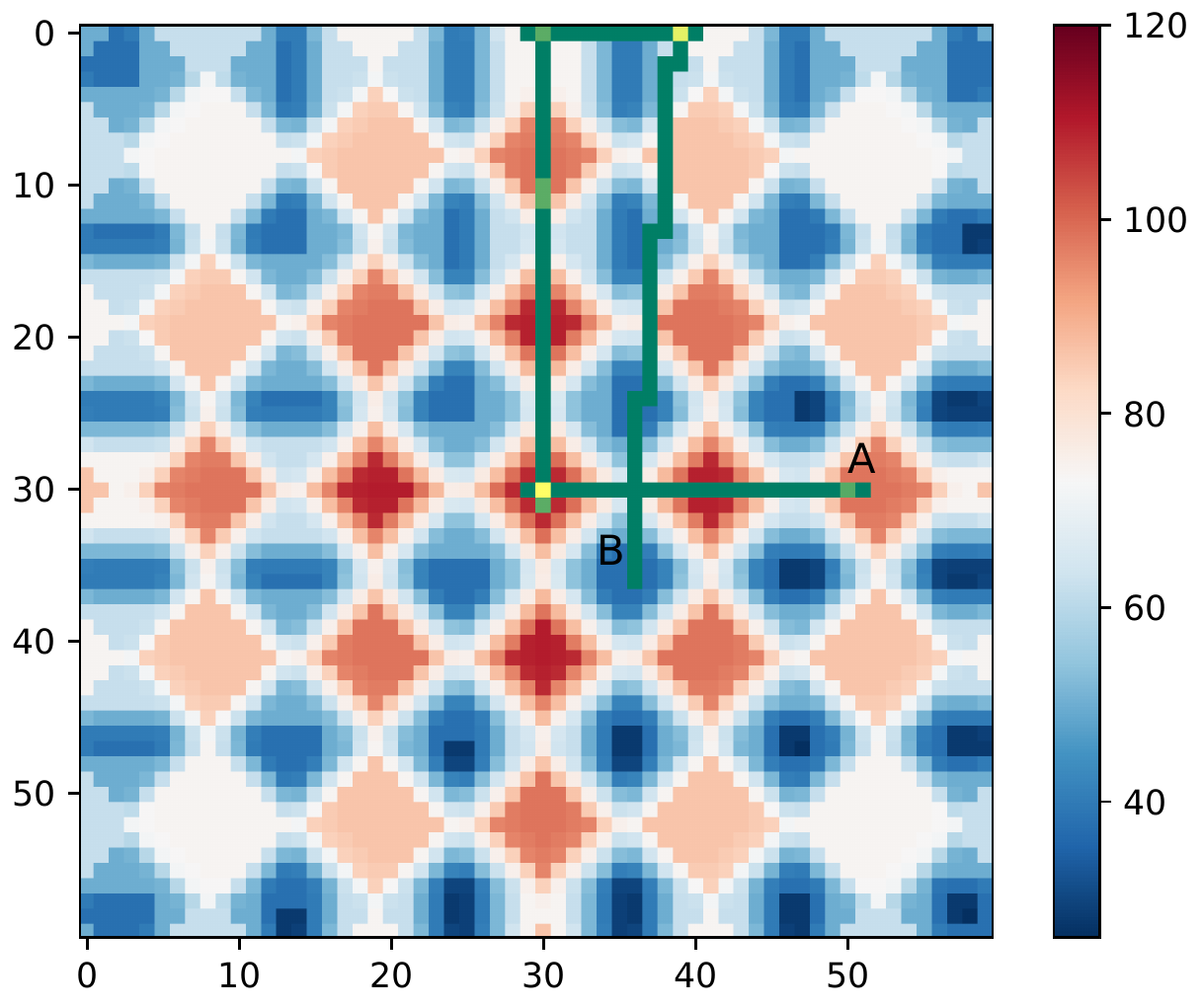}
  \caption{$L=28$}
  \label{fig:2}
\end{subfigure}\hfil 
\begin{subfigure}{0.24\textwidth}
  \includegraphics[width=\linewidth]{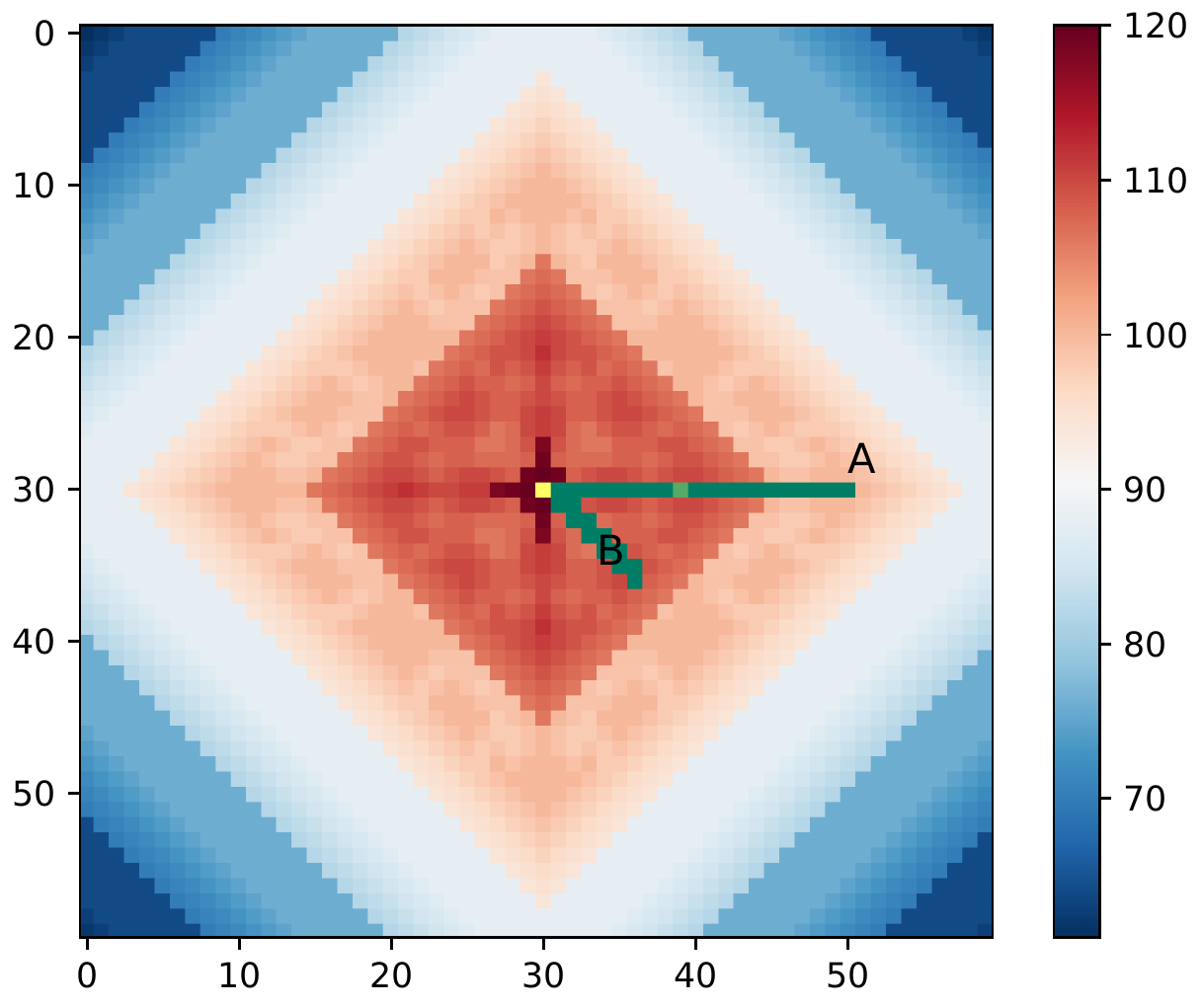}
  \caption{$L=30$}
  \label{fig:1}
\end{subfigure}\hfil 
\begin{subfigure}{0.24\textwidth}
  \includegraphics[width=\linewidth]{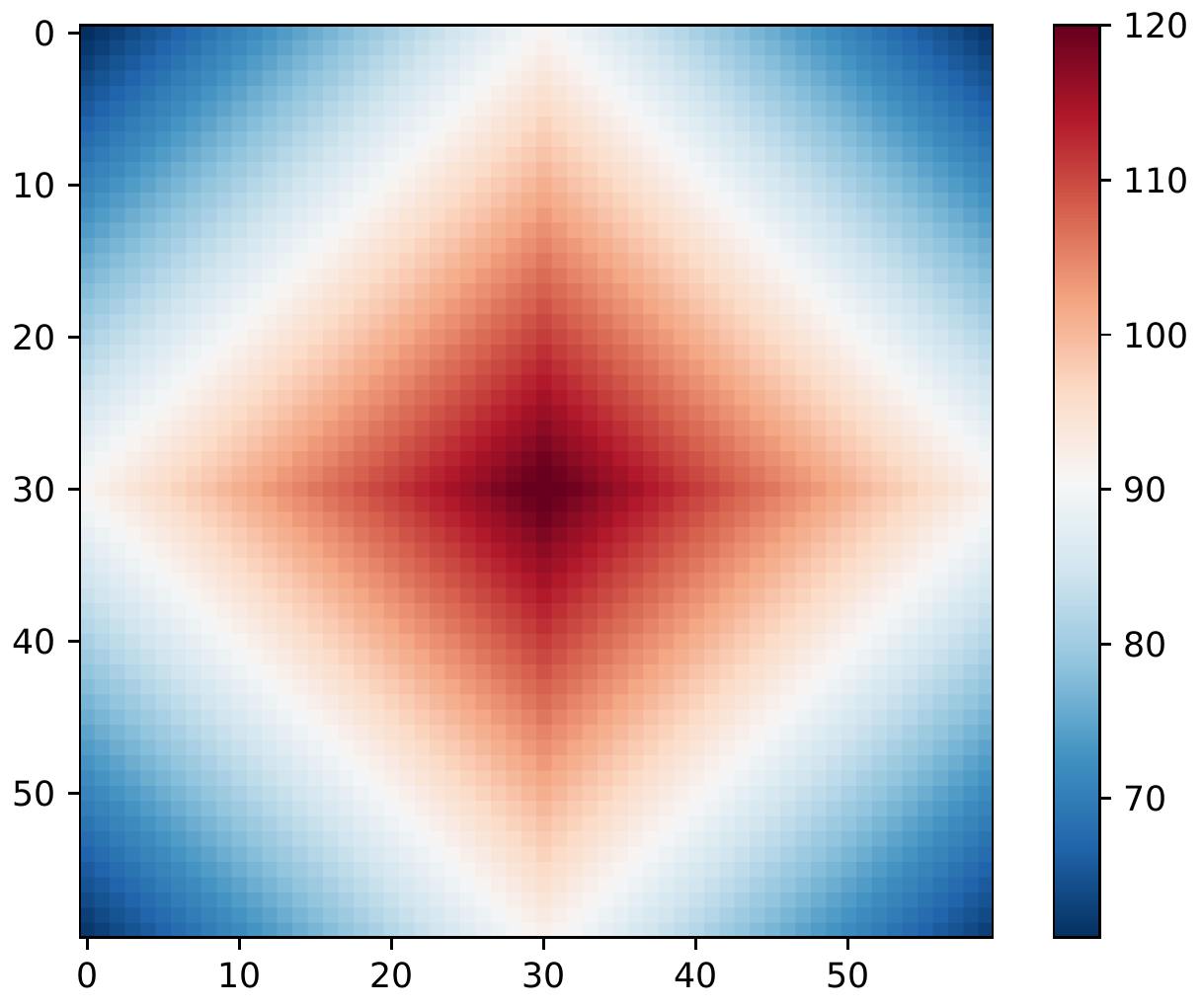}
  \caption{$L=\infty$}
  \label{fig:2}
\end{subfigure}
\caption{Same as Fig 2. except for Setup 2.
Shown are also two trajectories starting from A and B, obtained by Algorithm~\ref{Alg:local_dp}.}
\label{fig:value_functions_middle}
\end{figure*}

\subsection{Evaluation 2}\label{case_study_2}
We continue with SCAP, considering the hard-constrained version given by Algorithm \ref{Alg:local_dp}, calculating $\mathcal{A}_k$ by going through all sequences. The aim of the investigation is mainly to see how the complexity-performance tradeoff affects the behaviour. Towards this, we consider two similar setups, both considering the large room but where the first setup has the goal state $s^\star$ in the lower right corner of the room and the second setup has $s^\star$ in the middle of the room.
\subsubsection{Setup 1: $s^\star$ in the corner}
In Setup 1, we penalise each stage equally with complexity limit $L_k=L$, having stage length $l=12$, and consider the large room from Section \ref{case_study_1} with a slightly larger horizon $T = 119$ to be compatible with the stage partition (i.e., $l \cdot (K+1) = T+1$), setting $K=9$ accordingly. The change in tradeoff comes into play by varying $L$ with values $26, 28$ and 30. These values are picked since they illustrate the tradeoff well, allowing almost no action sequences for $L=26$ while many for $L=30$.
 

\subsubsection{Result} 

Value functions $V_0(s)$ are plotted in Fig. \ref{fig:value_functions_corner} for $L \in \{26,28,30,\infty\}$, where $L=\infty$ is as a reference obtained by ordinary dynamic programming maximising the total reward $\sum_{t=0}^T r_t(s_t,a_t)$. Blue (red) colour encodes low (high) value.

For $L=26$, the room is divided into squares of length $l=12$ due to the stage partition and the low complexity limit $L=26$. More precisely, at each stage $k$, $\mathcal{A}_k$ consists of only five action sequences: the ones executing only one action throughout the stage. This constraint implies that the robot sometimes has to wait at the walls, yielding the discontinuous jumps in the value function seen in Fig. \ref{fig:value_functions_corner} (a). 
Consider for example the two trajectories in Fig. \ref{fig:value_functions_corner} (a) generated by Algorithm \ref{Alg:local_dp}, coloured in a green-to-yellow scale, where states more visited are more yellow. The robot starting at A reaches $s^\star$ without waiting. The robot starting at B goes right two stages and then down two stages, but hits the wall at its second going-right stage before the stage is done. This causes waiting at the wall and thus longer time to reach $s^\star$, even though $B$ is actually closer to $s^\star$ than A.

For $L=28$, the squares in Fig. \ref{fig:value_functions_corner} (b) are greatly diminished compared to $L=26$, due to additional possible action sequences in $\mathcal{A}_k$. Here, a robot starting at B does not wait anymore at the wall. We note also that the squares are more evident near the goal state since the robot has less time here to adapt and avoid wall waiting. 

Finally, for $L=30$, $\mathcal{A}_k$ has increased so much that the difference with the unconstrained case is minor, and the trajectories from A and B, already being optimal at $L=28$, has not changed.

\subsubsection{Setup 2: $s^\star$ in the middle}
Setup 2 is identical to Setup 1, except that $s^\star$ is placed in the middle of the room. Thus, the robot can no longer heavily rely on the wall dynamics as in Setup 1. This notably changes the outcome.

\subsubsection{Result}
Similar to Setup 1, we plot $V_0(s)$ for different $L$ in Fig. \ref{fig:value_functions_middle}. 

For $L=26$, $V_0(s)$ has higher values in a grid-like pattern, with highest value in the middle of the room at $s^\star$ and high values at points being away a multiple of $l=12$ from $s^\star$. At all these points, the robot can (using the five sequences in $\mathcal{A}_k$) reach and stay at $s^\star$ resulting in a high total reward. Between those points are also states with higher value. Here, the robot can reach but not stay at $s^\star$; Instead, the robot oscillates back and forth crossing the $s^\star$ multiple times. This is seen for the robot starting at A in Fig. \ref{fig:value_functions_middle} (a). For the remaining states, the robot can never reach the goal state, due to the heavy complexity constraint. This is the case for the robot starting and staying at B. Thus, for low complexity limits, the robot may never reach the intended objective. 

As we increase $L$ to $28$, the grid-like pattern is expanded to groups of states due to the increased set of admissible action sequences. A robot starting at A reaches now and stays at $s^\star$ as can be seen in Fig. \ref{fig:value_functions_middle} (b). The robot starting at B also reaches and stays at $s^\star$, but takes a detour to be within the local complexity limit, reaching the upper wall and then down to $s^\star$. Thus, for low complexity limits, the robot may execute an action sequence which is locally of low complexity, but globally quite complex. 

Finally, for $L=30$, $V_0(s)$ is similar to the unconstrained case, except that the stage length $l=12$ together with the complexity constraint still partitions the space into entities of length $l=12$, similar to Fig. \ref{fig:value_functions_corner} (a) (but in this case into diamond-like shapes). The robots starting at A and B now reaches the goal fast with the high complexity limit.

\section{Conclusion}\label{conclusion}

In this paper, we have considered complexity-aware planning for DFA with rewards as outputs. We first defined a complexity measure for deterministic policies based on Kolmogorov complexity. Kolmogorov complexity is used since it can detect computational regularities, a typical feature for optimal policies. 

We then introduced a complexity-aware planning objective based on our complexity measure, yielding an explicit trade-off between a policy's performance and complexity. It was proven that maximising this objective is non-trivial in the sense that dynamic programming is infeasible. 

We presented two algorithms for obtaining low-complexity policies, COPS and SCAP. COPS finds a policy with low complexity among all optimal policies, following a two-step procedure. In the first step, the algorithm finds all optimal policies, without any complexity constraints. Dropping the complexity constraints, this step can be done using dynamic programming. The second step runs a uniform-cost search over all optimal policies, guided by a complexity-driven cost, favouring low-complexity policies. SCAP modifies instead the objective penalising policies locally for executing complex manoeuvres. This is done by partitioning the horizon into stages and enforce local complexity constraints over these stages, where the partition enables dynamic programming. We illustrated and evaluated our algorithms on a simple navigation task where a mobile robot tries to reach a certain goal state. Our algorithms yield low-complexity policies that concur with intuition.


Future work includes comparisons with other estimation methods of the Kolmogorov complexity and how uncertainty (e.g., stochasticity) and feedback (e.g., receding horizon) can be incorporated into the existing framework. Finally, another challenge is to tractably maximise the complexity-aware planning objective in the general case.


\bibliographystyle{plain}
\bibliography{Ref3}

\if\longversion1 

\section{Appendix}
\subsection{Proofs}
\subsubsection{Proof of the observation in Example \ref{ex_simple}}
We prove that \eqref{eq:alg_ec} and \eqref{eq:simple_optimal_policies} are equivalent for $\beta>0$ sufficiently small.
\begin{proof}
For brevity, let
\begin{equation*}
R(a_{0:T}) = \sum_{t=0}^T r_t(s_t,a_t)
\end{equation*} 
denote the total reward, subject to the dynamics $s_{t+1} = f_t(s_t,a_t)$ and start state $s_0 \in S$. Note that the statement is trivially true if $R(a_{0:T})$ is constant (over all $a_{0:T} \in A^{T+1}$), since \eqref{eq:alg_ec} and \eqref{eq:simple_optimal_policies} are then equivalent for all $\beta>0$. Hence, we may assume that $R(a_{0:T})$ is not constant. In particular, the difference between the maximum of $R(a_{0:T})$ and its second highest value is then positive:
\begin{equation*}
d:= \max_{a_{0:T} \in A^{T+1}} R(a_{0:T}) - \max_{a_{0:T} \in A^{T+1} \backslash \Omega} R(a_{0:T}) >0,
\end{equation*}
where $\Omega := \argmax_{a_{0:T} \in A^{T+1}} R(a_{0:T})$. 

We start by showing that a maximiser of \eqref{eq:alg_ec} is a minimiser of \eqref{eq:simple_optimal_policies} given that $\beta>0$ is sufficiently small (to be specified). Towards this, let $a^*_{0:T}$ be a maximiser of \eqref{eq:alg_ec}. Assume by contradiction that $a^*_{0:T} \notin \Omega$. Fix any $\tilde{a}_{0:T} \in \Omega$ and observe~that
\begin{align}
R(\tilde{a}_{0:T}) - R(a^*_{0:T}) \geq d.
\end{align}
Moreover, since $a^*_{0:T}$ is a maximiser of \eqref{eq:alg_ec}, we also have~that
\begin{align*}
[R(a^*_{0:T}) - \beta K(a^*_{0:T})] - [R(\tilde{a}_{0:T}) - \beta K(\tilde{a}_{0:T})] \geq 0,
\end{align*}
and therefore
\begin{align}\label{simple_optimal_policy_proof_eq1}
\beta [K(\tilde{a}_{0:T})-K(a^*_{0:T})] \geq R(\tilde{a}_{0:T}) - R(a^*_{0:T}) \geq d.
\end{align}
Note that \eqref{simple_optimal_policy_proof_eq1} does not hold for
\begin{equation}\label{simple_optimal_policy_proof_eq2}
\beta < \frac{d}{\max_{a_{0:T}} K(a_{0:T}) - \min_{a_{0:T}} K(a_{0:T})}.
\end{equation}
Thus, for $\beta>0$ sufficiently small specified by \eqref{simple_optimal_policy_proof_eq2}, a maximiser $a^*_{0:T}$ of \eqref{eq:alg_ec} must belong to $\Omega$, and since
\begin{align*}
\argmax_{a_{0:T} \in \Omega} [R(a_{0:T}) - \beta K(a_{0:T})] = \argmin_{a_{0:T} \in \Omega} K(a_{0:T}),
\end{align*}
we conclude that $a^*_{0:T}$ is a minimiser of \eqref{eq:simple_optimal_policies}. 

We now conversely show that a minimiser of \eqref{eq:simple_optimal_policies} is also a maximiser of \eqref{eq:alg_ec} given that $\beta>0$ is sufficiently small specified by \eqref{simple_optimal_policy_proof_eq2}. More precisely, let $a^*_{0:T}$ be a minimiser of \eqref{eq:simple_optimal_policies}. Let $\tilde{a}_{0:T}$ be any maximiser of \eqref{eq:alg_ec}. By above, $\tilde{a}_{0:T} \in \Omega$, hence $R(a^*_{0:T}) = R(\tilde{a}_{0:T})$, and since $K(a^*_{0:T}) \leq K(\tilde{a}_{0:T})$ we get
\begin{align*}
R(a^*_{0:T})-\beta K(a^*_{0:T}) \geq R(\tilde{a}_{0:T})- \beta K(\tilde{a}_{0:T}),
\end{align*}
from which we conclude that $a^*_{0:T}$ is a maximiser of \eqref{eq:alg_ec}. 

By above, we conclude that \eqref{eq:alg_ec} and \eqref{eq:simple_optimal_policies} are equivalent given that $\beta>0$ is sufficiently small specified by \eqref{simple_optimal_policy_proof_eq2}. This completes the proof.
\end{proof} 

\subsubsection{Proof of Proposition \ref{co:no_dyn_prog}}
Proposition \ref{co:no_dyn_prog} follows immediately from the following lemma.

\begin{lemma}\label{th:no_dyn_prog}
For $T \in \mathbb{N}$ sufficiently large, there do not exist functions $\{ h_i \}_{i=1}^T$ such that
\begin{equation}\label{eq:th4_main}
K(x_{1:T}) = \sum_{i=1}^T h_i(x_i)
\end{equation}
holds for all sequences $x_{1:T} \in A^{T}$.
\end{lemma}

\begin{proof}
We may without loss of generality restrict ourselves to the binary case $A = \{0,1\}$. We prove the lemma by contradiction. Let $\phi$ be the partial computable function in Definition \ref{definition_kolmogorov_complexity}. Let $T \in \mathbb{N}$ be arbitrary. By a simple counting argument, there exists at least one $x^*_{1:T}$ such that $K_\phi(x^*_{1:T}) \geq T$ (e.g., Theorem~2.2.1 in \cite{li2008introduction}). Also, the compliment $\tilde{x}_{1:T}$ of $x^*_{1:T}$, defined by inverting all zeros and ones in $x^*_{1:T}$, has complexity close to $x^*_{1:T}$ in the sense that
\begin{equation}\label{eq:th4_proof_1}
| K_\phi(\tilde{x}_{1:T}) - K_\phi(x^*_{1:T}) | \leq c_{\hat{M}}
\end{equation}
holds for some constant $c_{\hat{M}}$ independent of $x^*_{1:T}$ and $T$. To see this, let $M$ be the Turing machine that takes a binary string $p$, inverts all zeros and ones, and outputs the result $\tilde{p}$. In particular, $M(x^*_{1:T}) = \tilde{x}_{1:T}$.\footnote{For brevity, we use, for a given Turing machine $M$ with corresponding partial computable function $\psi$, the notation $M(p)$ to denote $\psi(p)$ for~${p \in A^*}$.} Let in turn $\hat{M}$ be the Turing machine that given input $p$ simulates the universal Turing machine $U$ corresponding to $\phi$, obtains the output $U(p)$ and then feeds it as input to $M$. Then $U(p) = x^*_{1:T}$ implies $\hat{M}(p) = \tilde{x}_{1:T}$. Hence, letting $\psi_{\hat{M}}$ be the corresponding partial computable function of $\hat{M}$, we have by the invariance~theorem, 
\begin{equation}\label{eq:th4_proof_2}
K_\phi(\tilde{x}_{1:T}) \leq K_{\psi_{\hat{M}}}(\tilde{x}_{1:T}) + c_{\hat{M}} \leq K_{\phi}(x^*_{1:T}) + c_{\hat{M}},
\end{equation}
where $c_{\hat{M}}$ is independent of $x^*_{1:T}$ and $T$. Furthermore, since $M(\tilde{x}_{1:T}) = x^*_{1:T}$, $U(p) = \tilde{x}_{1:T}$ implies $\hat{M}(p) = x^*_{1:T}$, we have again by the invariance theorem,
\begin{equation}\label{eq:th4_proof_3}
K_\phi(x^*_{1:T}) \leq K_\phi(\tilde{x}_{1:T}) + c_{\hat{M}}.
\end{equation}
Combining equation \eqref{eq:th4_proof_2} and \eqref{eq:th4_proof_3} yields \eqref{eq:th4_proof_1}. 

We now show that $K_\phi(0^T) \leq \log_2(T)+c_0$, for some constant $c_0$ independent of $T$, where $0^T$ is the string consisting of $T$ zeros. To see this, let $p_T \in A^*$ be the binary string corresponding to $T$ and note that $\ell(p_T) \leq \log_2(T)+c$, where $c$ is independent of $T$.\footnote{This correspondence between $A^*$ and $\mathbb{N}$ is given by the length-increasing lexicographic ordering where $x = 2^{n+1}-1+\sum_{i=0}^n a_i 2^i$ in $ \mathbb{N}$ corresponds to $a_n \dots a_1 a_0$ in $A^*$, see Chapter 1.4 in \cite{li2008introduction}.} Consider the Turing machine $N$ that given a binary string $p$ as input, calculates the corresponding integer $n$ and outputs $0^n$. In particular, $N(p_T)=0^T$. Let $\psi_N$ be the corresponding partial computable function of $N$. By the invariance theorem,
\begin{align}
K_\phi(0^T) \leq K_{\psi_N}(0^T)+c_N \leq \nonumber \\ \log_2(T)+c+c_N =: \log_2(T)+c_0.
\end{align}
Similarly, $K_\phi(1^T) \leq \log_2(T)+c_1$, for some constant $c_1$ independent of $T$, where $1^T$ is the string consisting of $T$~ones.

Assume now that \eqref{eq:th4_main} holds. Then,
\begin{align*}
2T-c_{\hat{M}} \leq K_\phi(x^*_{1:T})+K_\phi(\tilde{x}_{1:T}) = \sum_{i=1}^T h_i(0)+h_i(1) = \\ K_\phi(0^T)+K_\phi(1^T) \leq 2\log_2(T)+c_0+c_1.
\end{align*}
This yields a contradiction for $T$ sufficiently large, which proves the lemma.
\end{proof}

\subsubsection{Proof of Proposition \ref{th:search}}
We need the following lemma for uniform-cost search.
\begin{lemma}[\cite{aibook}]\label{lemma:search}
Assume the cost $c$ in a uniform-cost search is such that $c(n) \leq c(n')$ for any node $n$ and any child node $n'$ of $n$. Then the uniform-cost search terminates at a node with lowest cost.
\end{lemma}

\begin{proof}[Proof of Proposition \ref{th:search}]
The assumption in Proposition \ref{th:search} corresponds to the assumption in Lemma \ref{lemma:search}. Hence, the uniform-cost search terminates at an action sequence $a_{0:T}$ with the lowest execution complexity, and since $a_{0:T}$ is optimal by construction, the result follows.
\end{proof}

\subsection{Details of SCAP}
\subsubsection{Uniform-cost search for finding $\mathcal{A}_k$}
One can use a uniform-cost search for trying to find $\mathcal{A}_k$ for moderate stage length $l$. More precisely, the uniform-cost search is conducted with nodes $n=(t,a_{0:t-1})$, cost $c(n) = K(a_{0:t-1})$, and termination criteria $c(n) > L_k+\delta_k$ for some $\delta_k \geq 0$. If $K(a_{0:t-1}) \leq K(a_{0:t})$ holds in the search, then one can set $\delta_k=0$ and the search terminates when all sequences in $\mathcal{A}_k$ have been considered (cf. Proposition \ref{th:search}). However, since this assumption may not always hold throughout the search, one consider instead a positive $\delta_k>0$ as a margin (increasing the search, though without any formal guarantees of having found all elements of $\mathcal{A}_k$ at termination). 
\subsubsection{Action sequence extraction} The action sequence extraction for the hard-constrained version \eqref{eq:hard_con_version} is given by
\begin{equation*}
\bold{a}_k \in \argmax_{a_{0:l-1} \in \mathcal{A}_k} \left [ r^l_k(s_{lk},a_{0:l-1})+V_{k+1}(f^l_k(s_{lk},a_{0:l-1})) \right ]
\end{equation*}
for $k=0,1,\dots,K$, starting at $k=0$, with next state given by $s_{l(k+1)} = f^l_k(s_{lk},\bold{a}_k)$. In the end, at $k=K$, the total action sequence $a_{0:T} = (\bold{a}_0,\dots,\bold{a}_K)$ has been obtained. The procedure for \eqref{eq_modified_objective} is analogous.

\subsection{The Estimation Method}
We provide a brief overview of the method from \cite{zenil2018decomposition} used to estimate the Kolmogorov complexity $K(x)$, followed by a discussion concerning potential bias this method might generate in the numerical evaluations in Section \ref{numerical}.
\subsubsection{Overview} The method \cite{zenil2018decomposition} estimates $K(x)$ by first partitioning $x \in A^*$ into substrings, all of fixed block length $l>0$ (plus one possible remainder substring with length less than $l$). Denoting such a substring $x^i$ and letting $n_i$ be the number of occurrences of $x^i$ in the partition of $x$, the complexity $K(x)$ is then estimated by
\begin{equation}\label{eq:estimation_decomposition}
\tilde{K}(x) = \sum_i [k(x^i)+\log(n_i)],
\end{equation}
where $k(x^i)$ approximates $K(x^i)$ based on the coding theorem in algorithmic probability \cite{li2008introduction}, referring to \cite{zenil2018decomposition} for details. In words, $\tilde{K}(x)$ is the total complexity for generating all the substrings $x^i$, plus the code-length needed for specifying how frequent each $x^i$ is in $x$, which can be encoded in a codeword of length $O(\log(n_i))$. The method can thus be seen as a hybrid method leveraging both classical and algorithmic information theory. See \cite{zenil2018decomposition} for details and variants.

\subsubsection{Potential bias}\label{appendix_potential_bias}
We picked the default length $l=12$ as block length when using the estimation method in the numerical evaluations in Section \ref{numerical}. The decomposition in \eqref{eq:estimation_decomposition} makes the method focus on local computational regularities within the blocks $x^i$ and recurrence over the blocks. Consequently, the blocks $x^i$ found are intuitively simple, either alternating between two actions periodically (e.g., trajectory 1 in Fig. \ref{fig:test1_fig_rec_side20}) or batching the same actions together and execute such batches sequentially (e.g., trajectory 1 in Fig. \ref{fig:test1_fig_rec_side10}), while recurrences over the blocks is seen in e.g., trajectory 17 in Fig. \ref{fig:test1_fig_rec_side20}. With a method detecting global computational regularities, we might have seen sequences such as 1 and 3 in Fig. \ref{fig:test1_fig_rec_side10} (going to another corner first before reaching $s^\star$) for e.g., the medium room too. However, this is not the case since such a sequence for the medium size room has estimated complexity 92, significantly higher than the ones in Table \ref{table_2}. 

The remainder substring (with block length less than $l$) in the decomposition method may also generate somewhat different behaviour towards the end of the action sequence (e.g., 13-30 in Fig. \ref{fig:test1_fig_rec_side10}) because it is treated separately from the other $l$-length blocks due to its shorter length. This is most notable for the small room case, where the remainder part has a relatively significant part of the whole sequence length, causing some variation towards the end of the trajectories in Fig. \ref{fig:test1_fig_rec_side10}. For the medium size room and the large room, the variation has relatively a more minor effect. This remainder effect can of course be removed by adapting the horizon length to be divisible by $l$.

\fi

\end{document}